\documentclass[sigconf]{acmart}

\usepackage{balance}

\renewcommand\footnotetextcopyrightpermission[1]{} 

\usepackage{enumitem}
\usepackage{booktabs} 

\renewcommand{\C}{\mathcal{C}}
\renewcommand{\P}{\mathrm{P}}
\newcommand{\E}{\mathbb{E}}
\newcommand{\btheta}{\boldsymbol{\theta}}

\usepackage{amsmath}
\DeclareMathOperator*{\argmax}{arg\,max}

\newtheorem{prop}{Proposition}

\usepackage[ruled,vlined]{algorithm2e}

\usepackage{caption}
\usepackage{subcaption}

\copyrightyear{2019}
\acmYear{2019} 
\setcopyright{iw3c2w3}
\acmConference[WWW '19]{Proceedings of the 2019 World Wide Web Conference}{May 13--17, 2019}{San Francisco, CA, USA}
\acmBooktitle{Proceedings of the 2019 World Wide Web Conference (WWW '19), May 13--17, 2019, San Francisco, CA, USA}
\acmPrice{}
\acmDOI{10.1145/3308558.3313429}
\acmISBN{978-1-4503-6674-8/19/05}

\begin{document}
\title[Community Detection through Likelihood Optimization: In Search of a Sound Model]{Community Detection through Likelihood Optimization: \\ In Search of a Sound Model}
\titlenote{This is a preprint for the paper published in Proceedings of the 2019 World Wide Web Conference (WWW'19) under the Creative Commons Attribution 4.0 International (CC-BY 4.0) license. This version contains a few additional technical details.}

\author{Liudmila Prokhorenkova}
\affiliation{%
 \institution{Moscow Institute of Physics and Technology}
  \institution{Yandex}
  \city{Moscow}
  \country{Russia}
}
\email{ostroumova-la@yandex.ru}

\author{Alexey Tikhonov}
\affiliation{%
  \institution{Yandex}
  \city{Berlin}
  \country{Germany}
}
\email{altsoph@gmail.com}


\begin{abstract}
Community detection is one of the most important problems in network analysis. Among many algorithms proposed for this task, methods based on statistical inference are of particular interest: they are mathematically sound and were shown to provide partitions of good quality. Statistical inference methods are based on fitting some random graph model  (a.k.a. \textit{null model}) to the observed network by maximizing the likelihood. The choice of this model is extremely important and is the main focus of the current study. We provide an extensive theoretical and empirical analysis to compare several models: the widely used planted partition model, recently proposed degree-corrected modification of this model, and a new null model having some desirable statistical properties. We also develop and compare two likelihood optimization algorithms suitable for the models under consideration. An extensive empirical analysis on a variety of datasets shows, in particular, that the new model is the best one for describing most of the considered real-world complex networks according to the likelihood of observed graph structures. 
\end{abstract}

%
%


\maketitle

\section{Introduction}

Among various properties shared by many real-world complex networks, community structure is extremely important. It is characterized by the presence of highly interconnected groups of vertices (\textit{communities} or \textit{clusters}) relatively well separated from the rest of the network. In social networks communities are formed by users with similar interests, in citation networks they represent papers on related topics, etc. The presence of communities highly affects, e.g., the promotion of products via viral marketing, the spreading of infectious diseases,
computer viruses and information, and so on. 
Being able to identify communities in a network could help us to exploit this network more effectively: find similar scientific papers, discover users with similar interests for targeted advertisement, compress or visualize the network, etc.

Among other algorithms proposed for community detection, the notable ones are methods based on statistical inference. In such methods, some underlying random graph model is assumed, the evidence is represented by the observed graph, and hidden variables to be inferred are the parameters of the model together with community assignments. Such methods are appealing since they are theoretically sound and consistent: e.g., it has been proved that when the maximum-likelihood method is applied to networks generated from the same stochastic block model, it returns correct cluster assignments in the limit of large degrees~\citep{bickel2009nonparametric}. Also, likelihood can be used to formalize the notion of a community~\cite{copic2009identifying}. 

The choice of the proper null model is essential for statistical inference algorithms as it highly affects their performance. There are two types of models: with fixed or non-fixed number of parameters. Direct likelihood optimization for models with non-fixed number of parameters leads to trivial partitions, as discussed in Section~\ref{sec:likelihood_method}, and additional heuristics are needed to define the number of clusters. Hence, in this paper we focus on models with fixed number of parameters.  We analyze strengths and weaknesses of two most popular models~--- planted partition model and its degree-corrected variant (they are discussed in Section~\ref{sec:null_models}). We also propose a new model, which has only one parameter, satisfies a desired statistical property of preserving the expected degree sequence, and shows the best fit for a variety of real-worlds networks according to the likelihood of observed structures. 

Our research is similar in spirit to the recent paper by~\citet{yang2015defining}, where the authors provided an extensive comparison of various community scoring functions. In contrast, we focus solely on the likelihood optimization methods and, particularly, on the comparison of the null models. Note that such models allow to score partitions on the global level instead of individual communities. To the best of our knowledge, we are the first to perform such thorough evaluation of the null models used for community detection.

To sum up, the contributions of this paper are the following:
\begin{itemize}[noitemsep,nolistsep]
\item We theoretically analyze two well-known null models used for community detection. Based on this analysis, we propose a new one-parameter model which has a desirable statistical  property of preserving the expected degree sequence.
\item We empirically demonstrate that the new model gives the best fit to a variety of real-world complex networks. 
\item We show that the likelihood based on the proposed model can be used as a new, more reliable, alternative to the well-known modularity measure.
\item We develop several likelihood optimization algorithms suitable for all discussed models. We perform an extensive empirical comparison of all the algorithms on various synthetic and real-world datasets. 
\item The implementation of the proposed algorithms is available as an open-source library.\footnote{\url{https://github.com/altsoph/community_loglike}}
\end{itemize}

The rest of the paper is organized as follows. In the next section we overview related papers and introduce the required background. Then, in Section~\ref{sec:null_models}, we formally define and analyze all null models used in the current research. In Section~\ref{sec:likelihood_optimization}, we discuss the method of statistical inference and present the likelihoods for all models. The proposed likelihood optimization algorithms are discussed in Section~\ref{sec:optimization}. Section~\ref{sec:experiments} provides an extensive empirical analysis of all models and algorithms and Section~\ref{sec:conclusion} concludes the paper.

\section{Background and related work}\label{sec:related_work}

\subsection{Community Detection Methods}

In this section, we briefly overview algorithms used for community detection. However, most of the related work will be discussed in the corresponding background sections: Section~\ref{sec:modularity} discusses modularity measure; Section~\ref{sec:LFR} defines the LFR bechmark; Sections~\ref{sec:SBM} and~\ref{sec:DCSBM} describe standard null models; Section~\ref{sec:inference_related} discusses papers on statistical inference; Section~\ref{sec:optimization_related} deals with some optimization techniques. For a deeper overview of community detection area, we refer the reader to several survey papers, e.g.,~\citep{chakraborty2017metrics,coscia2011classification,fortunato2010community,fortunato2016community,malliaros2013clustering}.

The problem of community detection has recently attracted a lot of attention and many algorithms were proposed to solve this problem. The proposed methods include: spectral algorithms~\citep{von2007tutorial};  methods based on statistical inference (see Section~\ref{sec:likelihood_optimization}); methods based on optimization (see Section~\ref{sec:optimization}); methods based on dynamical processes on networks, like diffusion, spin dynamics, synchronization, and so on.  Also, existing approaches differ by a particular task at hand: detecting overlapping communities~\citep{yang2014structure}, clustering in dynamic~\citep{spiliopoulou2011evolution}, directed~\citep{rosvall2008maps}, or weighted~\citep{aicher2014learning} networks, community detection in annotated~\citep{newman2016structure} or layered~\citep{peixoto2015inferring} networks, and so on. 

In this paper, we focus on the basic problem of detecting non-overlapping communities in undirected and unweighted graphs. However, ideas discussed in this paper can be applied to more general cases, since there exist the corresponding generalizations of the null models. We leave this analysis for the future research, since there is a lot to be understood even for the basic case.

\subsection{Modularity}\label{sec:modularity}

\begin{small}
\begin{table}
  \caption{Notation}
  \label{tab:notation}
  \begin{tabular}{cp{5.1cm}}
    \toprule
    Variable&Description\\
    \midrule
    $G$ & observed graph, simple and undirected \\
    $V(G)$ & set of vertices \\
    $E(G)$ & set of edges \\
    $n$ = |V(G)| & number of vertices in $G$ \\
    $m = |E(G)|$ & number of edges in $G$ \\
    $\C = \{C_1, \ldots, C_k\}$ & partition of $V(G)$ \\
    $k = |\C|$ & number of communities \\
    $Q_{0}(\C,G,\gamma)$ & simple modularity defined in~\eqref{eq:simple_modularity} \\
    $Q_{1}(\C,G,\gamma)$ & standard modularity defined in~\eqref{eq:modularity} \\
    $c(i)$ & cluster assignment of vertex $i$ \\
    $m_{in} = m_{in}(\C,G)$ & number of intra-community edges \\ 
    $m_{out} = m_{out}(\C,G)$ & number of inter-community edges \\
    $P = \binom{n}{2}$ & number of pairs of vertices \\
    $P_{in} = \sum_{C \in \C} \binom{|C|}{2}$ & number of intra-community pairs of vertices \\
    $P_{out} = P - P_{in}$ & number of inter-community pairs of vertices \\
    $d(i)$ & degree of vertex $i$ in $G$ \\
    $d_{in}(i)$ & number edges going from $i$ to vertices in $C_{c(i)}$ \\
    $d_{out}(i) = d - d_{in}(i)$ & number of edges going from $i$ to outside $C_{c(i)}$ \\
    $D(C) = \sum_{i \in C} d(i)$ & degree of cluster $C$ \\
    $D_{in}(C) = \sum_{i \in C} d_{in}(i)$ & twice the number of edges induced by $C$ \\
    $m(C_q,C_r)$ & number of edges between $C_q$ and $C_r$ or twice the number of intra-cluster edges if $q=r$ \\
    $e(i,j)$ & number of edges between $i$ and $j$  \\
  \bottomrule
\end{tabular}
\end{table}
\end{small}

Throughout this paper we assume that we are given a simple undirected graph $G = (V(G),E(G))$ and the task is to find a partition $\C = \{C_1, \ldots, C_k\}$ of its vertex set $V(G)$ such that $\bigcup C_i = V(G)$ and $C_i \bigcap C_j = \emptyset$ for $i \neq j$. The notation used in the paper is summarized in Table~\ref{tab:notation}.

For various problems related to community detection it is extremely important to have a quality function $\bar Q(\C,G)$ which measures a goodness of a partition $\C$ for a graph $G$. Having such a function, one can:
\begin{itemize}[noitemsep,nolistsep]
\item Optimize $\bar Q(\C,G)$ to detect good communities according to this measure;
\item Use $\bar Q(\C,G)$ as a stopping criteria for a community detection algorithm (e.g., to choose the partition level in hierarchical algorithms);
\item Tune parameters of a parametric community detection algorithm;
\item Compare several community detection algorithms when no ground truth partition is available;
\item Compare several candidate partitions of $G$.
\end{itemize}

The most well-known such measure is 
\textit{modularity} which was first introduced by Newman and Girvan~\cite{newman2004finding}. 
The basic idea is to consider the fraction of intra-community edges among all edges of $G$ and penalize it in order to avoid trivial partitions like $\C = \{V(G)\}$. In its general form, modularity is
$$
Q(\C,G,\gamma) = \frac{1}{m} \left( m_{in} - \gamma \E (m'_{in}) \right)\,,
$$
where $m = m(G) = |E(G)|$ is the number of edges in $G$; $m_{in} = m_{in}(\C,G)$ is the number of intra-community edges; $m'_{in} = m'_{in}(\C)$ is a (random) number of intra-community edges in a graph constructed according to some underlying random graph model; $\gamma$ is the \textit{resolution parameter}~\cite{reichardt2006statistical}, which is usually added to get the possibility of varying the number of communities obtained after maximizing $Q(\C,G,\gamma)$.

The simplest underlying model is the Erd{\H{o}}s--R{\'e}nyi random graph, i.e., we assume that $G$ is sampled uniformly at random from the set of all graphs with $n$ vertices and $m$ edges. For this model, 
$\E (m'_{in}) =P_{in} \cdot \frac{m}{P}$,
where $P = \binom{n}{2}$ is the number of pairs of vertices and $P_{in} = \sum_{C \in \C} \binom{|C|}{2}$ is the number of intra-community pairs. So, the formula for \textit{simple modularity} is
\begin{equation}\label{eq:simple_modularity}
Q_{0}(\C,G,\gamma) = \frac{1}{m} \left( m_{in} - \gamma\frac{P_{in} m}{P} \right)\,.
\end{equation}

However, the Erd{\H{o}}s--R{\'e}nyi random graph model is known to be not a good descriptor of real-world networks since its Poisson degree distribution significantly differs from heavy-tailed degree distributions observed in real-world networks. A standard solution is to consider a random graph with a given expected degree sequence~\cite{chung2002connected} as a null model, and to take the degree sequence from the observed graph.\footnote{At large scale, this model is essentially equivalent to the configuration model~\cite{molloy1995critical}: a random graph sampled uniformly from the set of graphs with a given degree sequence.} For this model, the probability that vertices $i$ and $j$ are connected equals $\frac{d(i)d(j)}{2m}$, where $d(i)$ is the degree of a vertex $i$.\footnote{For technical purposes, it is more convenient (and conventional) to allow the underlying model to have loops and multiple edges and to assume that the number of edges between vertices $i$ and $j$ follows a Poisson distribution with the mean $\frac{d(i)d(j)}{2m}$ if $i \neq i$ and the expected number of loops for a vertex $i$ is $\frac{d(i)^2}{4m}$. In particular, in this case we avoid problems when $d(i)d(j) > 2m$ which can possibly occur. Note that usually $d(i)d(j) \ll 2m$, so multiple edges rarely appear.} In this case we have
$$
\E m'_{in} = \frac{1}{2m} \sum_{C \in \C} \frac{1}{2} \sum_{i \in C} \sum_{j \in C} d(i)d(j) = \frac{1}{4m} \sum_{C \in \C} D(C)^2,
$$
where $D(C) = \sum_{i \in C} d(i)$ is the degree of cluster $C$. So, the standard expression for \textit{modularity} is
\begin{equation}\label{eq:modularity}
Q_{1}(\C,G,\gamma) = \frac{m_{in}}{m} -  \frac{\gamma}{4 m^2} \sum_{C \in \C} D(C)^2\,.
\end{equation}
Originally, modularity was introduced with $\gamma = 1$, but is was proven to have a resolution limit~\cite{fortunato2007resolution}, i.e., it fails to detect communities smaller than a certain size; varying the resolution parameter helps to overcome this problem, but tuning $\gamma$ is challenging.

Despite some known shortcomings~\citep{fortunato2010community}, modularity remains the most popular way to measure the quality of a partition and many algorithms are based on direct modularity optimization (see Section~\ref{sec:optimization}). Also, as we discuss in Section~\ref{sec:likelihood_optimization}, under certain assumptions, modularity maximization is equivalent to likelihood optimization~\cite{newman2016community}.
In this paper we also propose a new quality function based on a more consistent one-parametric null model.

\subsection{LFR Model}\label{sec:LFR}

Let us now describe the LFR (Lancichinetti-Fortunato-Radicchi) model~\cite{lancichinetti2008benchmark}, which is the most widely used synthetic benchmark for comparison of community detection algorithms.

LFR generates a graph $G$ on $n$ vertices. The main parameters of the model are: the exponent of the power-law degree distribution~$\gamma$, the exponent of the power-law community size distribution $\beta$,  the average degree $\bar d$, and mixing parameter $\hat\mu$.
First, we generate the degrees of vertices by sampling them independently from the power-law distribution with exponent $\gamma$ and mean $\bar d$. Each vertex shares a fraction $1-\hat\mu$ of its edges with the vertices of its community and a fraction $\hat\mu$ with the other vertices of the network. The sizes of the communities are sampled from the power-law distribution with exponent $\beta$, such that the sum of all sizes equals $n$. Then, vertices are assigned to communities such that the internal degree of any vertex is less than the size of its community. Finally, the configuration model~\citep{molloy1995critical} with rewiring steps is used to construct a graph with the given degree sequence and with the required fraction of internal edges. See~\citep{lancichinetti2008benchmark} for the detailed description of this procedure.

\section{Null Models}\label{sec:null_models}

In this section, we analyze several random graph models having community structure and suitable for using as null models for likelihood optimization. 

\subsection{Stochastic Block Model}\label{sec:SBM}

The most well-known random graph model with community structure is the \textit{stochastic block model} (SBM)~\cite{holland1983stochastic}. In this model, the vertices are divided in $k$ clusters and for each pair of vertices $i,j$ we draw an edge between them with probability $p_{c(i),c(j)}$ independently of all other edges; here $c(i)$ is a community assignment for a vertex $i$. 
In other words, the probability of an edge between two vertices depends only on their community assignments.
The values $p_{q,r}$, $1 \le q,r \le k$ are parameters of the model, $0 \le p_{q,r} \le 1$ and $p_{q,r} = p_{r,q}$. 
The matrix of the probabilities $p_{q,r}$ is called the stochastic block matrix. If the diagonal elements are larger than the other ones, then generated graphs have community structure. 

In many community detection algorithms a simplified version of SBM called \textit{planted partition model} (PPM) is used~\cite{condon2001algorithms}. PPM is a special case of SBM, where $p_{q,q} = p_{in}$ for $1 \le q \le k$ and $p_{q,r} = p_{out}$  for $1 \le q,r \le k, q\neq r$, $p_{in}$ and $p_{out}$ are parameters of PPM. 

\subsection{Degree-Corrected Stochastic Block Model}\label{sec:DCSBM}

Similarly to the Erd{\H{o}}s--R{\'e}nyi random graph, SBM is unable to model heavy-tailed degree distributions, so community detection algorithms based on this model may have a poor quality. To overcome this issue, Karrer and Newman proposed the \textit{degree-corrected stochastic block model} (DCSBM)~\citep{karrer2011stochastic}. In this model the vertices are again assigned to $k$ clusters and edges are placed independently at random. The number of edges between vertices $i$ and $j$ follows a Poisson distribution with the mean $\frac{d(i)d(j)}{2m}p_{c(i),c(j)}$ or a half that number for self-loops. 
Let us show that DCSBM is able to generate graphs with desired expected degrees.
\begin{prop}\label{prop:positive}
There exist such $p_{q,r}$ that in DCPPM we have $\E (d'(i)) = d(i)$ for all~$i$.
\end{prop}
\begin{proof}
Let us set
$$
p_{q,r} = \frac{2\,m\,m(C_q,C_r)}{D(C_q)\, D(C_r)}\,,
$$
where $m(C_q,C_r)$ is the number of edges between clusters $C_q$ and $C_r$ or twice the number of intra-cluster edges if $q=r$. 
Note that these parameters maximize the likelihood for DCSBM~\citep{karrer2011stochastic}.
For such parameters we get the desired condition:
\begin{multline*}
\E (d'(i)) = \sum_{j=1}^n \frac{d(i)d(j)}{2m}p_{c(i),c(j)} \\
= \frac{d(i)}{D\big(C_{c(i)}\big)} \sum_{C\in\C}\frac{m\big(C,C_{c(i)}\big)}{D(C)} \sum_{j \in C} d(j) 
= d(i)\,.
\end{multline*}
So, we proved that in its general form DCSBM is able to preserve the desired degree sequence.
\end{proof}

The \textit{degree-corrected planted partition model} (DCPPM)~\cite{newman2016community} is a simplified version of DCSBM, where $p_{q,q} = p_{in}$ for $1 \le q \le k$ and $p_{q,r} = p_{out}$  for $1 \le q,r \le k, q\neq r$. Let us prove the following ``negative'' proposition.
\begin{prop}\label{prop:negative}
If $p_{in} \neq p_{out}$ and for some $q$, $r$ we have $D(C_q) \neq D(C_r)$, then $\E (d'(i)) \neq d(i)$ for some $i$.
\end{prop}
\begin{proof}
Let us compute the expected degree of $i$:

\begin{multline*}
\E (d'(i)) 
=  \frac{d(i)}{2m}p_{in} \sum_{j \in C_{c(i)}}  d(j) 
+   \frac{d(i)}{2m}p_{out}  \sum_{C \neq  C_{c(i)}}\sum_{j \in C} d(j) \\
=  \frac{d(i)}{2m} \left(  D\big(C_{c(i)}\big) (p_{in}-p_{out}) + 2m\,p_{out} \right)\,.
\end{multline*}
To get $\E (d'(i)) = d(i)$ for all $i$, we need to have, for \textit{any} cluster $C$,
$$
D(C) (p_{in}-p_{out}) + 2m\,p_{out} = 2m\,.
$$
As we have at least two clusters with different values $D(C)$, we have to take 
$
p_{in} = p_{out} = 1\,,
$
which leads to the standard configuration model without any community structure and contradicts the statement of the proposition. 
\end{proof}
As a result, we obtain that the standard method of limiting the number of parameters in DCSBM leads to inability of the obtained model to preserve the expected degree sequence.

\subsection{Independent LFR Model}

Motivated by Proposition~\ref{prop:negative}, we developed a one-parameter model which preserves the expected degree sequence.  It is a special case of DCSBM and we call it \textit{independent LFR model} (ILFR) due to its analogy to the LFR benchmark. 

The core problem of DCPPM is the assumption that the probability of an internal edge is independent of the size of its community, so the fraction of internal edges for a vertex \textit{depends on the community size}. On the contrary, we propose using the mixing parameter $\mu$ to control this fraction. Namely, we consider DCSBM with the following expected number of edges between two vertices $i$ and $j$: 
$$
\frac{\mu d(i)d(j)}{2m}\,\,\,\text{ if }c(i) \neq c(j)\,,
$$ 
$$
\frac{(1-\mu) d(i)d(j)}{D(C_{c(i)})} + \frac{\mu d(i)d(j)}{2m}\,\,\,\text{ if }c(i) = c(j)
$$ 
(or half this number for the self loops). Note that making the probability of an internal edge dependent on the size of the community is very natural: e.g., if a community of people is very small, one would expect that its members are much closer related to each other than members of large communities.
The following proposition holds for the proposed model.
\begin{prop}\label{prop:positice}
In ILFR we have $\E (d'(i)) = d(i)$ for all $i$.
\end{prop}
\begin{proof}
Indeed, let us compute the expected degree of $i$: 
$$
\E (d'(i)) = d(i)\left((1-\mu)\sum_{j\in C_{c(i)}} \frac{d(j)}{D(C_{c(i)})} + \mu \sum_{j \in V(G)} \frac{d(j)}{2m}\right) = d(i)\,.
$$
\end{proof}
So, this model solves the problem with expected degrees. Another advantage of ILFR is that it has only one parameter $\mu$ instead of $p_{in}$ and $p_{out}$ in planted partition models.

\section{Statistical inference}\label{sec:likelihood_optimization}
\subsection{Method}\label{sec:likelihood_method}

Having discussed possible options for the null models, we are ready to describe the statistical inference method. 
For community detection problem, this method consists of two ingredients: the evidence, expressed by a graph $G$, and a random graph null model with parameters $\btheta$ ($\btheta$ include community assignments and $\mu$ or $p_{in}$ and $p_{out}$ for the models under consideration). The goal is to find such parameters $\btheta$ that maximize the posterior distribution 
$$
\P(\btheta|G) = \frac{\P(G|\btheta) \P(\btheta)}{\P(G)}\,,
$$
where $\P(\btheta)$ is the prior distribution of parameters and
$
\P(G) = \int \P(G|\bar\btheta) \P(\bar\btheta) \, d \bar\btheta
$
is independent of $\btheta$, therefore 
$\argmax_{\btheta} \P(\btheta|G) = \argmax_{\btheta} \P(G|\btheta) \P(\btheta)\,.$
The choice of the prior distribution $\P(\btheta)$ is usually not obvious. Therefore, in the community detection literature, the likelihood $P(G|\btheta)$ that the  model is consistent with the observed graph structure is often maximized. 

Let us also stress that in Section~\ref{sec:null_models} we discussed two types of models: SBM/DCSBM having non-fixed number of parameters and PPM/DCPPM/ILFR with a fixed number of parameters. If the number of parameters is fixed, then the likelihood optimization can be applied directly. Otherwise, one needs another tool to find the optimal number of clusters, since direct likelihood maximization would lead to a trivial partition with all vertices forming their own clusters and the probability of an edge between two vertices equals 1 if they are connected and 0 otherwise. To avoid this problem, we further focus only on the models with fixed number of parameters. 

\subsection{Related Work on Statistical Inference}\label{sec:inference_related}

Let us briefly describe previous literature on statistical inference in community detection.
\citet{hastings2006community} considers PPM, as defined in Section~\ref{sec:SBM}. He shows a connection between the obtained likelihood and the Hamiltonian of a Potts model with short- and long-range
interactions. 
For $p_{in} > p_{out}$, the model is a spin glass with ferromagnetic nearest-neighbor interactions and antiferromagnetic long-range interactions. 
Belief propagation is used to find the ground state of the spin model. 
Another method based on the directed SBM is presented by~\citet{newman2007mixture}, they use the expectation-maximization technique to optimize the parameters.  
These ideas are close to the a posteriori blockmodel proposed in~\cite{nowicki2001estimation}. 
A variational approach to parameter estimation is adopted in~\citep{daudin2008mixture}. 
Some improvements and modifications of these approaches are introduced in, e.g.,~\citep{chen2015network,hofman2008bayesian,ramasco2008inversion,zanghi2008fast}. A generalization to the overlapping SBM is proposed by \citet{mcdaid2010detecting}.
As we discussed, DCSBM is defined in~\citep{karrer2011stochastic} and communities are detected based on the likelihood in this model. 
\citet{copic2009identifying} define an axiomatization for the problem of community detection based on PPM.
Likelihood is used as a quality function to define rankings between graph partitions and it is shown that such rankings satisfy a number of desired properties. 
The authors also propose an algorithm to find (approximately) the maximum likelihood partition.

Finally, if the number of parameters in the model is not fixed, some additional tools are required to figure out the optimal number of clusters. Here the average entropy of the classification~\citep{ramasco2008inversion}
or the Integrated Classification Likelihood~\citep{daudin2008mixture,zanghi2008fast} can be used.

\subsection{Likelihood for PPM and Simple Modularity}

The log-likelihood for PPM can be easily written:
\begin{multline*}
\log L'_{PPM}(\C,G,p_{in},p_{out}) = m_{in} \log p_{in} + m_{out} \log p_{out}  + \\ (P_{in}-m_{in})\log(1-p_{in}) 
 + (P_{out} - m_{out})\log(1-p_{out})\,.
\end{multline*}
Although it can be optimized directly using the algorithms discussed further in the paper, in order to demonstrate a connection of the log-likelihood to simple modularity defined in~\eqref{eq:simple_modularity} and to deal with all models in a unified way, we use the conventional trick and say that the number of edges between $i$ and $j$ follows a Poisson distribution with parameter $p_{in}$ or $p_{out}$.
Then we get:
\begin{multline*}
L_{PPM}(\C,G,p_{in},p_{out}) = 
\prod_{\substack{i,j: i < j,\\ c(i) = c(j)}} e^{-p_{in}} p_{in}^{e(i,j)} \\ \cdot
\prod_{\substack{i,j: i<j,\\ c(i) \neq c(j)}} e^{-p_{out}} p_{out}^{e(i,j)} 
= e^{-P_{in} p_{in}}  e^{-P_{out} p_{out}} p_{in}^{m_{in}} p_{out}^{m_{out}}\,,
\end{multline*} 
where $e(i,j)$ is the number of edges between $i$ and $j$, so
\begin{multline}\label{eq:likelihood_SBM}
\log L_{PPM}(\C,G,p_{in},p_{out}) = \\
m_{in} \log p_{in} + m_{out} \log p_{out}
 -P_{in} p_{in}  -P_{out} p_{out} \,.
\end{multline}
Note that the values $p_{in}$ and $p_{out}$ maximizing~\eqref{eq:likelihood_SBM} are
\begin{equation}\label{eq:opt_param_SBM}
p_{in} = \frac{m_{in}}{P_{in}}, \,\,\,\, p_{out} = \frac{m_{out}}{P_{out}}\,.
\end{equation}

In a resent paper by Newman~\citep{newman2016community} it was shown that if we assume $p_{in}$ and $p_{out}$ to be fixed, then maximizing~\eqref{eq:likelihood_SBM} is equivalent to maximizing simple modularity~\eqref{eq:simple_modularity}. Indeed,
\begin{multline*}
\log L_{PPM}(\C,G,p_{in},p_{out})  
 = m_{in} (\log p_{in} - \log p_{out}) + m \log p_{out} \\ - P_{in} (p_{in}  - p_{out}) -  P \, p_{out}
=  m \log p_{out} -  P \, p_{out} \\ + (\log p_{in} - \log p_{out})  \left( m_{in}  - P_{in} \frac{p_{in}  - p_{out}}{\log p_{in} - \log p_{out}}\right) \,.
\end{multline*}
So, we get the equivalence with 
\begin{equation}\label{eq:gamma_SBM}
\gamma = \frac{P(p_{in}  - p_{out})}{m(\log p_{in} - \log p_{out})}\,.
\end{equation}

\subsection{Likelihood for DCPPM and Modularity}\label{sec:likelihood_DCSBM}

Let us compute the log-likelihood for the DCPPM~\citep{newman2016community}:
\begin{multline}\label{eq:likelihood_DCSBM} 
\log L_{DCPPM}(\C,G,p_{in},p_{out}) \\ 
= - \frac 1 2 \sum_{\substack{i,j: \,c(i) = c(j)}} \frac{d(i)d(j)p_{in}}{2m} 
- \sum_{\substack{i,j:\, i< j,\\ c(i) \neq c(j)}} \frac{d(i)d(j)p_{out}}{2m} \\ 
+ \sum_{\substack{(i,j) \in E(G), \\i<j, c(i) = c(j)}} \log\frac{d(i)d(j)p_{in}}{2m} 
+ \sum_{\substack{(i,j) \in E(G), \\ i<j, c(i) \neq c(j)}} \log\frac{d(i)d(j)p_{out}}{2m} \\
 = m_{in} (\log p_{in} - \log p_{out})
 - \frac{p_{in} - p_{out}}{4m}  \sum_{C\in \C} D(C)^2
  \\ 
  + \sum_{i} d(i) \log d(i) 
 + m \log p_{out} 
 - m p_{out} - m \log (2m)\,.
\end{multline}
The values of parameters maximizing this likelihood are:
\begin{equation}\label{eq:opt_param_DCSBM}
p_{in} = \frac{4 \,m\,m_{in}}{\sum_C D(C)^2}\,,\,\,\,\, 
p_{out} = \frac{4 \,m \, m_{out}}{4 m^2 - \sum_{C}D(C)^2}\,.
\end{equation}

As shown in~\citep{newman2016community}, if $p_{in}$ and $p_{out}$ are fixed, then maximizing $\log L_{DCPPM}(\C,G,p_{in},p_{out})$ is exactly equivalent to maximizing modularity~\eqref{eq:modularity} with
\begin{equation}\label{eq:gamma_DCSBM}
\gamma = \frac{p_{in}  - p_{out}}{\log p_{in} - \log p_{out}}\,.
\end{equation}

\subsection{Likelihood for ILFR}


Let us compute the likelihood for the proposed ILFR model:
\begin{multline}\label{eq:likelihood_ILFR}
\log L_{ILFR}(\C,G,\mu) 
= - \frac 1 2 \sum_{\substack{i,j: c(i) = c(j)}} d(i)d(j)\left( \frac{(1-\mu)}{D(C_{c(i)})} + \frac{\mu}{2m} \right) 
\\ - \sum_{\substack{i,j: i< j,\\ c(i) \neq c(j)}} \frac{\mu d(i)d(j)}{2m} 
+ \sum_{\substack{(i,j) \in E(G), \\ i < j, c(i) = c(j)}} \log \left( d(i)d(j)\left( \frac{(1-\mu)}{D(C_{c(i)})} + \frac{\mu}{2m} \right) \right) \\
+ \sum_{\substack{(i,j) \in E(G), \\ i < j, c(i) \neq c(j)}} \log\frac{\mu d(i)d(j)}{2m} 
= 
 \sum_{C \in \C} \frac{D_{in}(C)}{2}\log \left( \frac{(1-\mu)}{D(C)}+ \frac{\mu}{2m} \right) \\
+  m_{out} \log \mu + \sum_i d(i) \log d(i) - m_{out} \log 2m - m \,,
\end{multline}
where $D_{in}(C) = \sum_{i \in C} d_{in}(i)$ is twice the number of edges induced by $C$. Note that the optimal value of $\mu$ is hard to find analytically, but it can be approximated numerically by optimizing~\eqref{eq:likelihood_ILFR}.

The following approximation helps to speed up the optimization algorithms and to make formulas more concise:
\begin{multline*}
\log \left( \frac{(1-\mu)}{D(C)}+ \frac{\mu}{2m} \right) = \log \frac{(1-\mu)}{D(C)} + \log\left( 1  + \frac{\mu\, D(C)}{(1-\mu)\,2\,m}\right) \\
\approx \log \frac{(1-\mu)}{D(C)} + \frac{\mu D(C)}{(1-\mu)\, 2\, m}
\approx \log \frac{(1-\mu)}{D(C)}\,,
\end{multline*}
since 
$
\left|\log \frac{1-\mu}{D(C)}\right| \gg \frac{\mu D(C)}{(1-\mu)\, 2\, m}\,.
$
This leads to another quality function, which we further refer to as ILFRS (S stands for ``simplified''):
\begin{multline}\label{eq:likelihood_ILFRS}
\log L_{ILFRS}(\C,G,\mu)  
=
m_{in} \log(1-\mu) +  m_{out} \log \mu  - m_{out} \log 2m \\
-  \sum_{C} \frac{D_{in}(C)}{2}\log D(C)    
 + \sum_i d(i) \log d(i)  -m \,.
\end{multline}
The optimal $\mu$ according to \eqref{eq:likelihood_ILFRS} can now be computed analytically:
\begin{equation}\label{eq:mu}
\mu = \frac{m_{out}}{m}\,.
\end{equation}

Note that now we can substitute $\mu$ in~\eqref{eq:likelihood_ILFRS} by the optimal value~\eqref{eq:mu} and obtain a non-parametric quality function: 
\begin{multline}\label{eq:ILFRS}
\log L_{ILFRS}(\C,G) = m_{in} \log\frac{m_{in}}{m} +  m_{out} \log \frac{m_{out}}{m}  -m  \\ - m_{out} \log 2m - \sum_{C} \frac{D_{in}(C)}{2}\log D(C)    
 + \sum_i d(i) \log d(i) \,.
\end{multline}
The obtained function is fairly simple and, as we show by further experiments, it can successfully replace the standard modularity function in many applications.

Let us also note that in contrast to ILFR, the likelihood for the standard LFR model cannot be computed. First, LFR is based on the configuration model, which introduces complex dependences between all edges, making the likelihood intractable. Second, and most importantly, the number of iter-community edges for each vertex is deterministic (the fraction of such edges is $\mu$); as a result, for most of graphs and partitions the likelihood is equal to zero.

\section{Optimization}\label{sec:optimization}

\subsection{Related Work on Modularity Optimization}\label{sec:optimization_related}

First, let us discuss some optimization approaches used in community detection. The most widely used measure to be optimized is modularity~\eqref{eq:modularity}. The following types of approaches are known.

\paragraph{Greedy optimization}
\citet{newman2004finding} proposed a greedy algorithm for modularity optimization, where the clusters are iteratively merged by choosing the pair which leads to a greater modularity increase (while it is positive). A speedup of the greedy algorithm was proposed by \citet{clauset2004finding}. Some other modifications were also suggested in, e.g., \citep{schuetz2008efficient,wakita2007finding}. Probably the most well-known and widely used greedy algorithm is called Louvain~\citep{blondel2008fast}. At the beginning each vertex forms its own cluster. Then we create the first level of the partition by iterating through all vertices: for each vertex $i$ we compute the gain in modularity coming from putting $i$ to the community of its neighbor and pick the community with the largest gain, as long as it is positive. After that the first level is formed and we replace the obtained communities by supervertices, and two supervertices are connected by a (weighted) edge if there is at least one edge between vertices of the corresponding communities. Then the process is repeated with the supervertices, and so on, as long as modularity increases. This algorithm was shown to be fast and provide partitions of good quality. We choose Louvain algorithm as the basis for our experiments, the detailed description of its application to our problem is given in Section~\ref{sec:alternatives}.

\paragraph{Simulated annealing}
Simulated annealing was first applied to modularity optimization in~\cite{guimera2005functional}. This method combines two types of moves: local moves, where a single vertex is shifted from one cluster to another; global moves, consisting of merges and splits of communities. Methods based on simulated annealing are relatively slow and cannot be applied to large datasets.

\paragraph{Spectral optimization}
Spectral methods are quite popular in modularity optimization. Such methods use the eigenvalues and eigenvectors of a so-called modularity matrix. 
Some examples can be found in~\citep{newman2006finding,sun2009improved}.

\vspace{5pt}
Many other algorithms exist, among them are methods based on extremal optimization~\citep{duch2005community}, mathematical programming~\citep{agarwal2008modularity}, mean field annealing~\citep{lehmann2007deterministic}, genetic algorithms~\citep{tasgin2007community}, etc.

\subsection{Proposed Likelihood Optimization Methods}\label{sec:alternatives}

In Section~\ref{sec:likelihood_optimization}, we presented four quality functions based on the likelihood: PPM~\eqref{eq:likelihood_SBM}, DCPPM~\eqref{eq:likelihood_DCSBM}, ILFR~\eqref{eq:likelihood_ILFR}, and ILFRS~\eqref{eq:likelihood_ILFRS}. All functions are parametric: PPM and DCPPM have parameters $p_{in}$ and $p_{out}$, ILFR and ILFRS have one parameter $\mu$. In this section, we discuss the possible ways to maximize these likelihoods. Below we assume that we are given an arbitrary algorithm $FindPartition(\bar Q, G)$ which is able to find a partition by maximizing some quality function~$\bar{Q}(\C,G)$. 

Although the optimization strategies proposed below are able to work with any algorithm $FindPartition$, let us first discuss our particular choice. As we already mentioned, we choose Louvain algorithm~\citep{blondel2008fast}, which is arguably the most widely used method for community detection. Louvain is fast and allows to process large datasets. Importantly, it can be adapted to all quality functions discussed in this paper. Initially the method is designed to optimize modularity, so it can be directly applied to the likelihood of DCPPM due to their equivalence. To adapt Louvain to other likelihoods, we have to change one part of the original algorithm: when we try to remove a vertex (or supervertex) from one community an add to another, we need to efficiently compute the difference in likelihood. This can be done for any of the considered quality functions: for ILFR and ILFRS the differences can be computed using~\eqref{eq:likelihood_ILFR} and~\eqref{eq:likelihood_ILFRS}; 
for PPM we have to additionally store the sizes of supervertices in order to compute the difference for $P_{in}$ in~\eqref{eq:simple_modularity}.
The rest of the Louvain algorithm remains the same.\footnote{To make our results reproducible, we made the source code publicly available at \url{https://github.com/altsoph/community_loglike}.} Therefore, the asymptotic time complexity of our Louvain-based $FindPartition$ and of the original Louvain are the same, up to a constant multiplier (the complexity of Louvain is empirically evaluated to be $O(n\log n)$).

Note that both the original Louvain and our modifications can be applied to the likelihoods with fixed values of parameters $p_{in}$, $p_{out}$ or $\mu$. Therefore, the key question is how to find the optimal values of these parameters.

\paragraph{Iterative strategy}  
The simplified version of this strategy was initially proposed in~\citep{newman2016community} for PPM and DCPPM models. Detailed description of our method is presented in Algorithm~\ref{alg:iterative}. Initially, we fix $\gamma=1$ or $\mu=0.5$ depending on the null model. Then, we apply \textit{FindPartition} to the corresponding log-likelihood function. Using the obtained partition we can re-estimate the parameters and continue this procedure until convergence or until some maximum allowed number of steps is reached. 
In our experiments we noticed that for any null model the parameters may end up cyclically varying near some value. Therefore, at each iteration we additionally check if we obtain an already seen parameter and stop in this case. This additional trick allows to significantly reduce the overall number of iterations. In our experiments we observed that we never need more than 50 iterations and the mean value is much smaller, see Section~\ref{sec:experiments} for more details.
\begin{algorithm}
\SetKwInOut{Input}{input}\SetKwInOut{Output}{output}
	\Input{\,\,\,graph $G$, $Model$, algorithm \textit{FindPartition}, $N$, $\varepsilon$ }
    \BlankLine
    Initialize $param = 1$ ($\gamma$ for PPM, DCPPM) or $param = 0.5$ ($\mu$ for ILFR, ILFRS), $Params = \emptyset$\;
    \For{$i \leftarrow 1$ \KwTo $N$}{
        Define $\bar{Q}$ according to~\eqref{eq:simple_modularity} for PPM, \eqref{eq:modularity} for DCPPM, \eqref{eq:likelihood_ILFR} for ILFR, \eqref{eq:likelihood_ILFRS} for ILFRS with parameter $param$\;
        $\C = FindPartition(\bar{Q},G)$\;
        \If{$Model$ is PPM or DCPPM}
        {Compute $p_{in}$ and $p_{out}$ according to~\eqref{eq:opt_param_SBM} or~\eqref{eq:opt_param_DCSBM}\;
        Compute $param_{new}$ according to~\eqref{eq:gamma_SBM} or~\eqref{eq:gamma_DCSBM}\;}
        \If{$Model$ is ILFRS}
        {Compute $param_{new}$ according to~\eqref{eq:mu}\;}
        \If{$Model$ is ILFR}
        {Compute $param_{new}$ by optimizing~\eqref{eq:likelihood_ILFR}\;}
         \If{$|param - param_{new}| < \varepsilon$ or $param_{new} \in Params$}
    {break\;}
     Add $param_{new}$ to $Params$\;
     $param \leftarrow param_{new}$\;
    }
     \Return{$\C$}
	{\caption{Iterative strategy}
    \label{alg:iterative}}
\end{algorithm}
\paragraph{Maximization strategy.} An alternative strategy is to directly search for parameters which maximize likelihood of the obtained partition. Here we can use any black-box optimization method.
At each iteration of optimization, for some value of $\gamma$ or $\mu$, we run $FindPartition$ to obtain a partition $\C$, then, using $\C$, we find the optimal parameters according to~\eqref{eq:opt_param_SBM}, \eqref{eq:opt_param_DCSBM},~\eqref{eq:mu} or by optimizing~\eqref{eq:likelihood_ILFR}. Finally, we compute the likelihood to be maximized using~\eqref{eq:likelihood_SBM}, \eqref{eq:likelihood_DCSBM}, \eqref{eq:likelihood_ILFR}, or \eqref{eq:likelihood_ILFRS}.
In our experiments, for simplicity and reproducibility, we use grid search in this strategy. 
However, any other method of black-box parameter optimization can be used instead, e.g., random search~\cite{bergstra2012random},
Bayesian optimization~\cite{snoek2015scalable}, Gaussian processes~\cite{Golovin2017Googlevizier}, sequential model-based optimization~\cite{bergstra2011algorithms,hutter2011sequential}, and so on.

To summarize, the maximization strategy directly maximizes the obtained likelihood, while the iterative one searches for a ``stable'' parameter, i.e., the one which does not change after applying $FindPartition$. Note that the convergence of the iterative strategy is not guaranteed for real-world networks. Therefore, initially we expected the maximization strategy to be more stable. However, our experiments show that these two strategies demonstrate similar performance, but the iterative one is faster.

Let us mention another possible optimization strategy, which we do not use in our experiments. Given a partition $\C$, for each likelihood function except ILFR, we can compute the optimal values of parameters according to~\eqref{eq:opt_param_SBM}, \eqref{eq:opt_param_DCSBM}, and~\eqref{eq:mu}. Therefore, for the corresponding models we can substitute these parameters and obtain non-parametric likelihoods, as we did in~\eqref{eq:ILFRS} for ILFRS. 
Namely, for PPM we can replace $p_{in}$ and $p_{out}$ in~\eqref{eq:likelihood_SBM} by~\eqref{eq:opt_param_SBM} and obtain
$$
\log L_{PPM}(\C,G) = 
m_{in} \log \frac{m_{in}}{P_{in}} + m_{out} \log \frac{m_{out}}{P_{out}}
 - m \,.
$$
For DCPPM, we replace $p_{in}$ and $p_{out}$ in~\eqref{eq:likelihood_DCSBM} by~\eqref{eq:opt_param_DCSBM} and get
\begin{multline*}
\log L_{DCPPM}(\C,G)
 = m_{in} \log \frac{m_{in} \left(4 m^2 - \sum_{C}D(C)^2 \right)}{m_{out}\sum_C D(C)^2}  \\
 - \left(m_{in} -  \frac{m_{out}\sum_{C} D(C)^2}{4 m^2 - \sum_{C}D(C)^2} \right) 
  + \sum_{i} d(i) \log d(i) \\
 + m \log \frac{4 \,m \, m_{out}}{4 m^2 - \sum_{C}D(C)^2}
- \frac{4 \,m^2 \, m_{out}}{4 m^2 - \sum_{C}D(C)^2} - m \log (2m).
\end{multline*}
All obtained non-parametric quality functions can potentially be optimized directly since they do not have free parameters.
However, we do not consider such strategy since: 1) it cannot be applied to ILFR as there is no analytical formula for optimal $\mu$, 2) this strategy cannot be easily combined with the Louvain algorithm. 
The reason is that Louvain has several partition levels and partitions obtained on earlier levels cannot be changed later. At the beginning of the algorithm our estimates of $\mu$, $p_{in}$, and $p_{out}$ are far from optimal: all vertices form their own communities, so $\mu = 1$, $p_{out} = 1$, and $p_{in} = 0$. As a result, the first level of the partition, which has a big impact on the final quality, is constructed based on non-optimal parameters. In fact, in most cases the algorithm does not even start optimization due to the huge overestimate of $\mu$ or $p_{out}$.

\section{Experiments}\label{sec:experiments}

In this section, we conduct an extensive experimental study to compare the discussed models. We start with some preliminaries and discuss evaluation metrics and datasets. In Section~\ref{sec:compare_loglike}, we measure how well the models describe various real-world networks. Next, we compare all proposed community detection algorithms. Finally,  Section~\ref{sec:negative} presents some negative results on the limits of statistical inference algorithms when applied to real-world networks. 

\subsection{Evaluation Metrics}\label{sec:metrics}

In order to evaluate the performance of any algorithm, we have to compare the partition $\C$ obtained by this algorithm  with the ground truth partition $\C_{GT}$. The problem of choosing a good similarity measure for this comparison does not have any standard solution in community detection literature. Different similarity measures can possibly give preferences to different algorithms. 
Usually, a suitable measure is chosen according to a practical problem at hand. That is why in this work we compute and compare several standard similarity measures. 

In particular, we use Normalized Mutual Information (NMI) of two partitions $\C$ and $\C_{GT}$, which is often used for the comparison of community detection algorithms~\cite{bagrow2008evaluating,fortunato2010community}. The idea behind NMI is that if two partitions are similar, one needs very little information to infer $\C_{GT}$ given $\C$. Assume that cluster labels for $\C_{GT}$ and $\C$ are values of two random variables $\xi_{GT}$ and $\xi$. Then, 
$$
\mathrm{NMI}(\C_{GT},\C) = \frac{2\,I(\xi_{GT},\xi)}{H(\xi_{GT})+H(\xi)}, 
$$
where $I(\xi_{GT},\xi) = H(\xi) - H(\xi_{GT}|\xi)$ is the mutual information of $\xi_{GT}$ and $\xi$, 
$H(\xi)$ is the Shannon entropy of $\xi$,  
$H(\xi_{GT}|\xi)$ is the conditional
entropy of $\xi_{GT}$ given $\xi$ . 

We also use two well-known similarity measures based on counting correctly and incorrectly classified pairs of vertices. Let $n_{11}$ denote the number of pairs of vertices which are in the same community in both partitions $\C$ and $\C_{GT}$, $n_{01}$ ($n_{10}$) the number of pairs which are in the same community in $\C$ ($\C_{GT}$) and in different communities in $\C_{GT}$ ($\C$), and $n_{00}$ the number of pairs that are in different communities in both partitions. By combining the introduced values one can obtain several similarity measures. We use two most popular ones, the first is Rand index~\citep{rand1971objective}: $\frac{n_{11}+n_{00}}{n_{11}+n_{10}+n_{01}+n_{00}}$, i.e., the fraction of the number of correctly classified pairs of vertices to the total number of pairs. One problem of the Rand index is that its value is usually close to 1, since $n_{00}$ is typically very high. The well-known Jaccard index does not suffer from this problem, it is defined as 
$\frac{n_{11}}{n_{11}+n_{10}+n_{01}}$, i.e.,
the fraction of the number of vertex pairs classified in the same cluster in both partitions to the number of pairs classified in the same cluster in at least one partition.
\subsection{Datasets}\label{sec:datasets}

\subsubsection{Synthetic networks}\label{sec:synthetic}
We use the LFR model described in Section~\ref{sec:LFR}
to generate synthetic networks. The parameters of the generated graphs are the following: the number of vertices $n = 10^4$;
the parameter of the power-law degree distribution $\gamma = 2.5$; the average degree $\bar d = 30$; the parameter of the community size distribution $\beta = 1.5$, with the minimum cluster size 50 and the maximum 600; the mixing parameter $\hat\mu$  is varied in range $[0,1]$.

We additionally experimented with synthetic graphs of different edge densities obtained by varying the average degree. The results are omitted since the only difference observed is that the qualities of all algorithms are usually larger for denser graphs. Finally, note that there are other benchmark models proposed in the literature. However, community detection in synthetic networks is not the main focus in this paper, therefore we do not consider such models.


\urldef\snapurl\url{http://snap.stanford.edu/data/index.html#communities}

\begin{small}
\begin{table}
  \caption{Real-world datasets; $\mu$, $\gamma_0$, and $\gamma_1$ are computed for ground truth partitions according to~(\ref{eq:mu}), (\ref{eq:opt_param_SBM}, \ref{eq:gamma_SBM}) and (\ref{eq:opt_param_DCSBM}, \ref{eq:gamma_DCSBM})}
  \label{tab:datasets}
  \begin{tabular}{lcccccc}
    \toprule
Dataset&$n$&$m$&$k$&$\mu$&$\gamma_1$&$\gamma_0$\\
    \midrule
    Karate club~\citep{zachary1977information} 
    	& 34 & 78 & 2 & 0.128 & 0.78 & 0.78 \\
    Dolphins~\citep{lusseau2003bottlenose} 
    	& 62 & 159 & 2 & 0.038 & 0.54 & 0.55 \\
	Football~\citep{newman2004finding} 
    	& 115 & 613 & 11 & 0.325 & 2.39 & 2.57 \\
        	Political books~\citep{newman2006modularity} & 105 & 441 & 3 & 0.159 & 0.86 & 0.89  \\
	Political blogs~\citep{adamic2005political} 
    	& 1224 & 16715 & 2 & 0.094 & 0.72 & 0.72 \\
    email-Eu-core~\citep{leskovec2007graph} & 986 & 16064 & 42  & 0.664 & 2.74 & 2.80 \\
    Cora citation~\citep{vsubelj2013model} & 24166 & 89157 & 70 & 0.458 & 5.46 & 6.21 \\
AS~\citep{boguna2010sustaining} & 23752 & 58416 &   176 & 0.561 & 1.15 & 1.40 \\
  \bottomrule
\end{tabular}
\end{table}
\end{small}

\subsubsection{Real-world networks}\label{sec:real}

We collected several networks with different structural properties (see Table~\ref{tab:datasets}). 
In addition to widely used networks, such as Zachary's karate club, dolphin social network, and American college football, we also used annotated books about politics\footnote{V. Krebs, unpublished, \url{http://www.orgnet.com/}} and political blogs~\citep{adamic2005political}. 
The dataset email-Eu-core is obtained from SNAP,\footnote{\url{http://snap.stanford.edu/data/email-Eu-core.html}} here the ground truth communities correspond to the departments of a research institute. In Cora citation dataset communities correspond to subjects of research papers, while in  AS~\citep{boguna2010sustaining} the vertices are annotated with their countries.
These are all publicly available datasets we found for non-overlapping community detection.\footnote{It is easier to find datasets with ground truth overlapping communities, see, e.g., \snapurl} 

\begin{small}
\begin{table}
  \caption{Log-likelihoods: datasets with ground truth}
  \label{tab:likelihoods_gt}
  \begin{tabular}{lccc}
    \toprule
    Dataset&$\log L_{PPM}$&$\log L_{DCPPM}$&$\log L_{ILFR}$\\
    \midrule
    Karate & -206.12 &	-168.65 &	{\bf-168.63}\\
    Dolphins & -483.50 & -439.52 &	{\bf-428.64}\\
	Football & {\bf-1384.1} & -1426.7 & -1428.4 \\
    	Political books& -1363.8 & {\bf-1235.0} & -1243.3 \\
	Political blogs&-73912 & -50756 & {\bf-50750} \\
    Eu-core&-65559 & -48783 & {\bf-48483} \\
    Cora&-678306 & -593358 & {\bf-584730} \\
    AS&-542952 & -351537 & \bf{-329784} \\
  \bottomrule
\end{tabular}
\end{table}
\end{small}

\subsection{Comparison of Likelihoods}\label{sec:compare_loglike}

Since the main aim of this paper is to analyze and compare several null models having community structure, we start with the following research question: \textit{Which model gives the best fit for real-world complex networks}? To answer this question, we compare the probabilities that real-world networks were generated by each of the null models. Namely, we compared the log-likelihoods~\eqref{eq:likelihood_SBM} for PPM, ~\eqref{eq:likelihood_DCSBM} for DCPPM and~\eqref{eq:likelihood_ILFR} for ILFR.\footnote{We intentionally do not consider synthetic networks in this experiment: likelihoods are expected to be heavily affected by the particular synthetic model used. For example, LFR benchmark could give a preference to ILFR null model.} Note that~\eqref{eq:likelihood_ILFRS} is a simplified expression for~\eqref{eq:likelihood_ILFR}, so we consider it only as a quality function but not as log-likelihood. 

In the first experiment, we took the datasets described in Section~\ref{sec:real} and assumed that the partitions are defined by the ground truth cluster assignments provided for these datasets. Then, we computed the optimal parameters $p_{in}$, $p_{out}$ or $\mu$ and used them to compute the corresponding log-likelihood (see Table~\ref{tab:likelihoods_gt}). One can see that PPM is the best model describing the Football dataset, DCPPM is the best for Pol-books, while for all other datasets ILFR has the largest likelihood. Note that ILFR has only one parameter to be tuned while DCPPM has two. Therefore, we initially expected that for many datasets DCPPM may have a larger likelihood. In this case, to decide which model better describes the data, we would have to adapt some information criterion (e.g., the Bayesian information criterion) to our problem, which is a nontrivial task for the models under consideration. Surprisingly, for most of the datasets ILFR has a larger likelihood, which clearly indicates that this model is more suitable for describing real-world networks.\footnote{Note that using these results we cannot compare ILFR and DCPPM with PPM since ILFR and DCPPM are based on the observed degree sequences, while PPM only uses the cluster assignments.}

\begin{small}
\begin{table}
  \caption{Log-likelihoods: datasets without ground truth}
  \label{tab:likelihoods}
  \begin{tabular}{lcccc}
    \toprule
    Dataset&n&$\log L_{PPM}$&$\log L_{DCPPM}$&$\log L_{ILFR}$\\
    \midrule
    Karate &34& -191.182 &	-163.990 &	{\bf-160.154}\\
    Dolphins &62& -417.240 &-398.718 &	{\bf -394.538}\\
	Football &115& {\bf-1364.38} & -1407.43 & -1406.79 \\
    Political books&105& -1182.90 & -1135.57 & \bf-1088.77 \\
	Political blogs&1224&-60653.8 & -49912.3 & {\bf -49702.3} \\   
    Eu-core&986&-57421.1 & -46020.9 & {\bf-45469.2} \\
    Cora&24K&-512556 & -450154 & {\bf-425463} \\
    AS&24K& -449088 & -244745 & \bf{-227917} \\
    Ca-GrQc &5242& -50258.2 & -53393.9 & \bf-42259.3 \\
    Ego-Facebook &4039&	-241708 & -234311 & \bf-207910 \\
    p2p-Gnutella09 &8114& -167995 & -148594 & \bf-144711 \\
    Wiki-vote &7115& -504395 & \bf-388044 & -388380 \\
    Email-Enron &37K & -1192893 & -862315 & \bf-803161 \\
    Soc-Epinions1 &76K& -2825138 & -2087117 & \bf-2000906 \\
    Soc-Slashdot0811 &77K& -4011399 & -3120963 & \bf-2949833 \\
    ego-Twitter &81K& -6979664 & -5781209 & \bf-5314782 \\
  \bottomrule
\end{tabular}
\end{table}
\end{small}

However, even though for each dataset under consideration we have a ground truth partition (a.k.a. attributes or metadata), this partition can be not an ideal division of vertices into communities and also can be not the only ground truth partition possible. For example, many complex networks have hierarchical community structure: e.g., for Cora dataset we can use original attributes (\textit{/Artificial\_Intelligence/Machine\_Learning/Probabilistic\_Methods/}), or second-level  ones (\textit{/Artificial\_Intelligence/Machine\_Learning/}), or just first-level (\textit{Artificial\_Intelligence}).\footnote{By measuring the likelihood for all such partitions of Cora, we notices that original attributes provide the highest likelihood for all null models, therefore we further analyze only this partition.} We further discuss this problem in Section~\ref{sec:results_real}.

Fortunately, likelihoods can be compared even for datasets without any knowledge about community structure, which allows to compare null models on a much larger variety of datasets. 
In order to measure the likelihood for any graph, we first have to find a partition that maximizes this value (among all possible partitions), which is feasible only for very small datasets. However, the optimal partition (and the corresponding likelihood) can be approximately found by applying the corresponding maximization algorithm: PPM-max, DCPPM-max or ILFR-max. We performed such comparison on the datasets introduced in Section~\ref{sec:real} as well as on 8 new datasets of various nature downloaded from \url{http://snap.stanford.edu/data/index.html} (see Table~\ref{tab:likelihoods}). According to this experiment, 
for almost all datasets the largest likelihood is again obtained for ILFR model, despite it has fewer free parameters, which means that this model is the best one for describing real-world datasets, which is the main empirical result supporting the introduction of ILFR and the corresponding log-likelihood quality function. 

Let us note that the log-likelihoods presented in Table~\ref{tab:likelihoods} can be underestimated, since we cannot guarantee that our optimization algorithms find exactly the maximum likelihood partition. However, we believe that this does not introduce any bias into this experiment (i.e., does not change the conclusions) since the same optimization procedure is applied for all models. Also, note that in all cases the log-likelihoods in Table~\ref{tab:likelihoods_gt} are smaller than the corresponding ones in Table~\ref{tab:likelihoods}, which means that our optimization algorithms found partitions providing better likelihoods for the corresponding models than the ground truth ones, which is expected (since partitions were tuned).

\subsection{Community Detection Algorithms}\label{sec:compare_algorithms}

In this section, we compare all algorithms proposed in Section~\ref{sec:alternatives}. Note that we intentionally use only methods based on the likelihood optimization since this paper focuses on the analysis and comparison of null models.

\subsubsection{Synthetic networks}

\begin{figure}
\centering
  \centering
  \includegraphics[width=\linewidth]{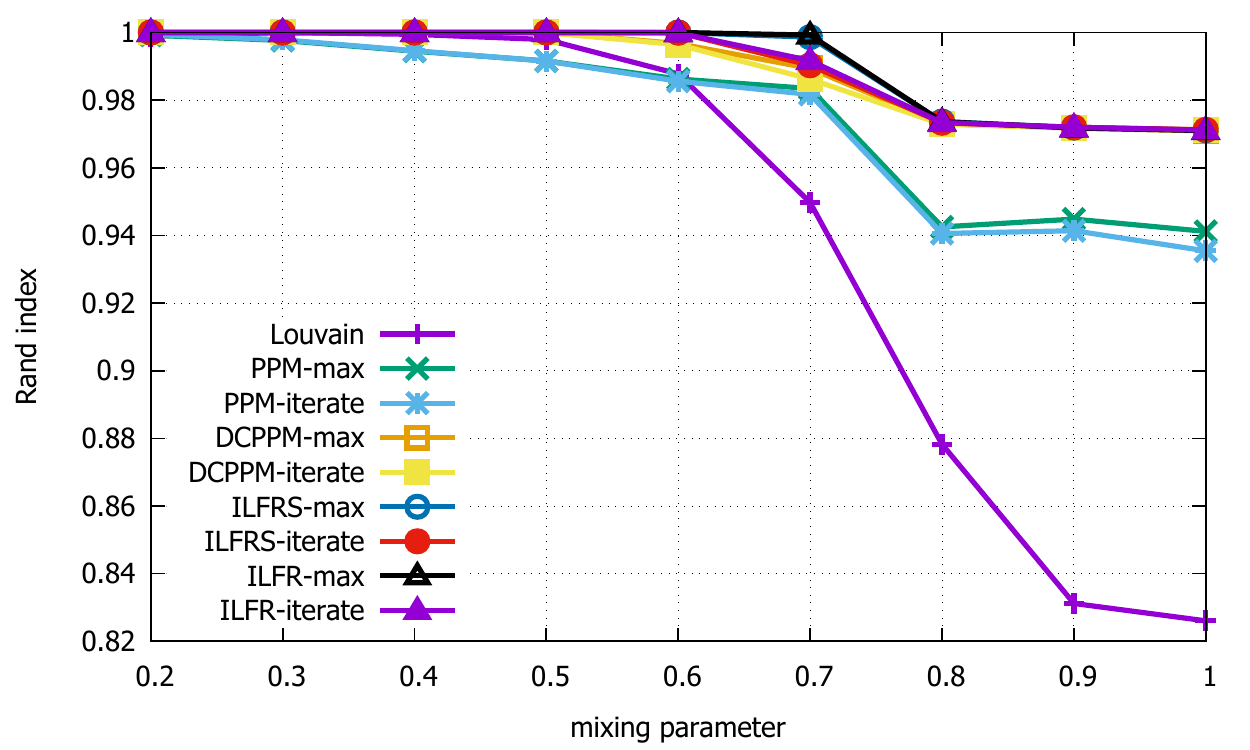}
  \includegraphics[width=\linewidth]{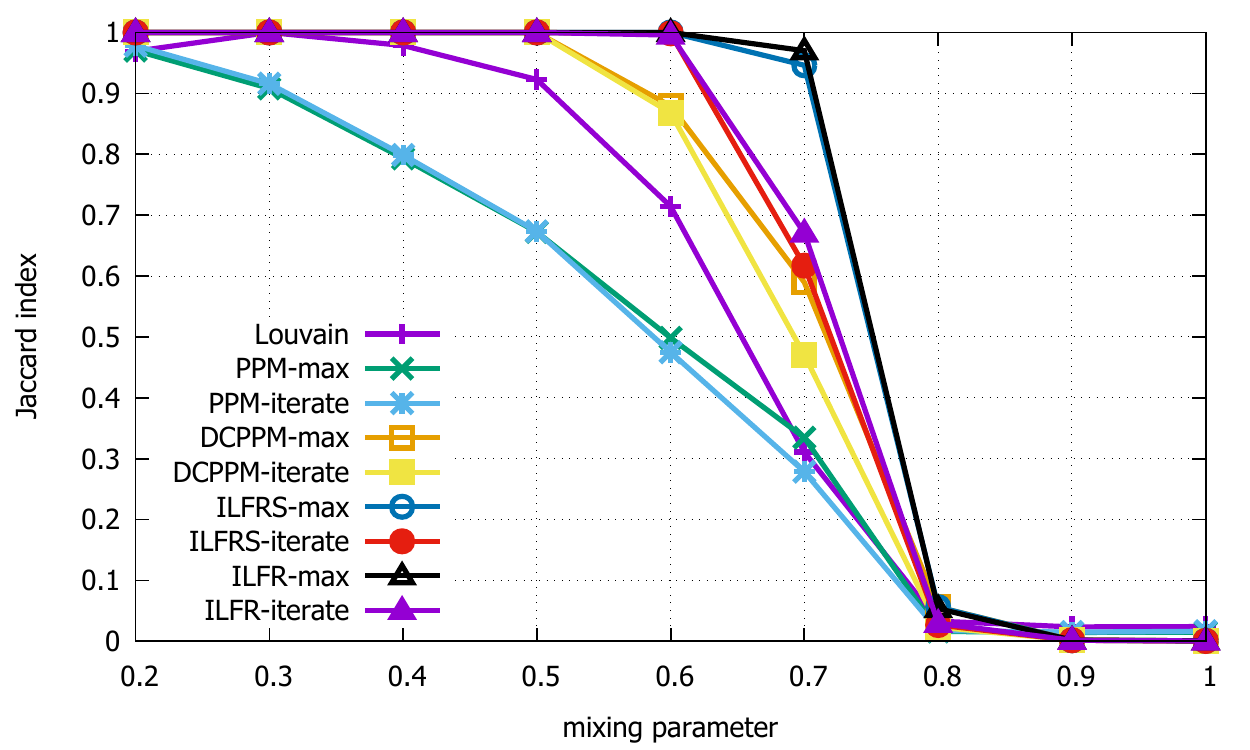}
  \includegraphics[width=\linewidth]{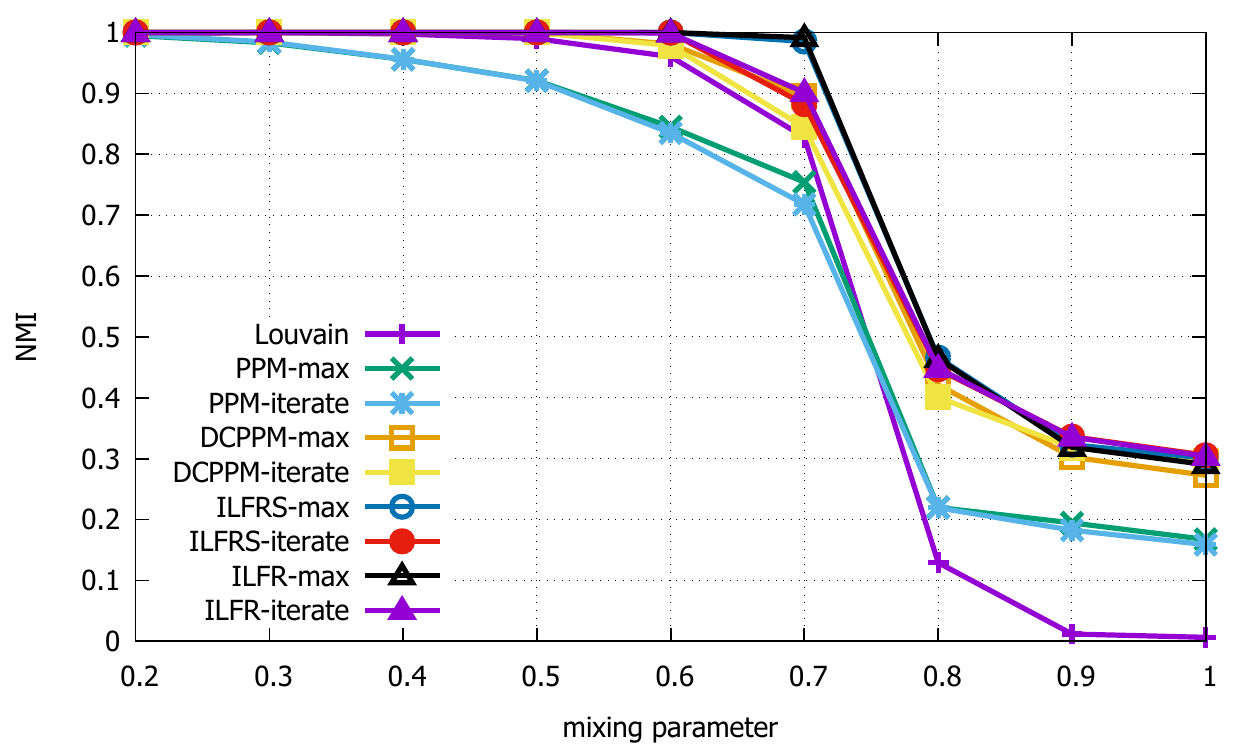}
  \caption{Comparison on synthetic networks}
  \label{fig:synthetic}
\end{figure}

\begin{table*}
  \caption{Maximization strategies (Rand index / Jaccard index / NMI) }
  \label{tab:results_real_max}
  \begin{tabular}{lccccc}
    \toprule
    Dataset&Louvain& PPM & DCPPM & ILFR & ILFRS \\
    \midrule
    Karate & 
    0.732 / 0.470 / 0.586 &
     0.707 / 0.397 / 0.585 &
     \textbf{0.766 / 0.523 / 0.667} &
     \textbf{0.774} / \textbf{0.535} / \textbf{0.687} &
     \textbf{0.774} / \textbf{0.535} / \textbf{0.687} \\
    Dolphins & 
    \textbf{0.634} / \textbf{0.351} / \textbf{0.499} & 
    0.529 / 0.153 / 0.386 &
    0.600 / 0.286 / 0.483 &
    0.594 / 0.274 / 0.472 &
    0.584 / 0.256 / 0.454 \\
	Football & 
    0.970 / 0.720 / 0.922 & 
    \textbf{0.993} / \textbf{0.916} / \textbf{0.972} & 
   \textbf{0.993} / \textbf{0.916} / \textbf{0.972} & 
    \textbf{0.993} / \textbf{0.916} / \textbf{0.972} & 
    \textbf{0.993} / \textbf{0.916} / \textbf{0.972} \\
    Pol-books & 
    0.831 / 0.616 / 0.545 & 
    0.722 / 0.331 / 0.449 & 
    0.817 / 0.583 / 0.536 &
    0.719 / 0.347 / 0.437 &
    0.772 / 0.477 / 0.490 \\
	Pol-blogs & 
    0.886 / 0.788 / 0.641 & 
    0.574 / 0.153 / 0.234 &
    \textbf{0.904} / \textbf{0.823} / \textbf{0.705} &
    0.860 / 0.738 / 0.570 &
    0.879 / 0.778 / 0.625 \\
    Eu-core & 
    0.869 / 0.218 / 0.584 & 
    0.942 / 0.207 / 0.661 &
    \textbf{0.964} / \textbf{0.434} / 0.740 & 
    0.953 / 0.371 / 0.727 &
    0.960 / 0.410 / 0.739  \\
    Cora & 
   0.943 / \textbf{0.127} / 0.460 & 
     \textbf{0.978} / 0.044 / 0.547& 
    0.978 / 0.050 / 0.533&
    0.978 / 0.063 / \textbf{0.551}&
    0.978 / 0.060 / \textbf{0.551}\\
   AS & 
   0.821 / \textbf{0.198} / \textbf{0.489} & 
   0.826 / 0.009 / 0.422 & 
   \textbf{0.826} / 0.026 / 0.461 &
    0.826 / 0.018 / 0.438 &
   0.826 / 0.018 / 0.438 \\
  \bottomrule
\end{tabular}
\end{table*}

\begin{table*}
  \caption{Iterative strategies (Rand index / Jaccard index / NMI)  }
  \label{tab:results_real_iterative}
  \begin{tabular}{lccccc}
    \toprule
    Dataset&Louvain& PPM & DCPPM & ILFR & ILFRS\\
    \midrule
    Karate & 
    0.732 / 0.470 / 0.586 & 
    0.743 / 0.480 / 0.612 & 
    0.750 / 0.502 / 0.612 & 
    0.766 / 0.523 / 0.667 &
    0.760 / 0.519 / 0.633 \\
    Dolphins & 
    \textbf{0.634} / \textbf{0.351} / 0.499 & 
    0.548 / 0.190 / 0.403 & 
    0.600 / 0.286 / 0.481 & 
    0.576 / 0.238 / 0.449 &
    0.575 / 0.238 / 0.448 \\
	Football & 
    0.970 / 0.720 / 0.922 & 
    \textbf{0.993} / \textbf{0.916} / \textbf{0.972} & 
     \textbf{0.993} / \textbf{0.916} / \textbf{0.972} & 
     \textbf{0.993} / \textbf{0.916} / \textbf{0.972} &
    \textbf{0.993} / \textbf{0.916} / \textbf{0.972} \\
    Pol-books & 
    \textbf{0.831} / \textbf{0.616} / \textbf{0.545} & 
    0.740 / 0.377 / 0.464 & 
    \textbf{0.843} / \textbf{0.651} / \textbf{0.554} & 
    0.760 / 0.445 / 0.481 &
    0.770 / 0.475 / 0.489  \\
    Pol-blogs & 
    0.886 / 0.788 / 0.641 & 
    0.582 / 0.171 / 0.240 & 
    \textbf{0.902} / \textbf{0.820} / \textbf{0.694} & 
    0.862 / 0.743 / 0.576 &
    0.877 / 0.773 / 0.618 \\
    Eu-core & 
    0.869 / 0.218 / 0.584 & 
    0.939 / 0.205 / 0.656 &
    \textbf{0.965} / \textbf{0.439} / \textbf{0.742} & 
    0.948 / 0.351 / 0.713 &
   0.955 / 0.387 / 0.725 \\
    Cora & 
    0.943 / \textbf{0.127} / 0.460 & 
    \textbf{0.978} / 0.049 / 0.547 & 
    0.978 / 0.052 / 0.534 & 
    0.978 / 0.058 / \textbf{0.549} &
    0.978 / 0.061 / \textbf{0.550} \\
    AS & 
    0.821 / \textbf{0.198} / \textbf{0.489} & 
    0.826 / 0.011 / 0.426 &   
    \textbf{0.826} / 0.025 / 0.460 & 
    0.826 / 0.017 / 0.436 &
    0.826 / 0.017 / 0.435 \\
  \bottomrule
\end{tabular}
\end{table*}

First, we compared all algorithms on synthetic networks described in Section~\ref{sec:synthetic}. For each $\hat \mu$ we generated 5 random samples of LFR and averaged the results.
Figure~\ref{fig:synthetic} shows the results obtained for all similarity measures under consideration (Rand, Jaccard and NMI). Note that in most cases these similarity measures are consistent, i.e., for a given mixing parameter they rank the algorithms similarly. However, Jaccard index is less descriptive for large mixing parameters.

Based on the obtained results, several important observations can be made.
For all values of $\hat\mu$, the best (or very close to the best) quality is obtained by ILFRS-max and ILFR-max strategies. 
For both ILFR and ILFRS it turns out that maximizing the likelihood is better than iteratively estimating the parameters, which supports our intuition that better partitions are expected to have higher likelihoods. 
For small values of $\hat\mu$ the worst results are obtained by PPM (both strategies), while for large values Louvain algorithm is the worst.

\begin{table*}
  \caption{Quality functions: ground truth vs optimization}
  \label{tab:compare_likelihood}
  \begin{tabular}{lcccccccccc}
    \toprule
   &
    \multicolumn{2}{c}{Modularity} &
    \multicolumn{2}{c}{$\log L_{PPM}$} &  
    \multicolumn{2}{c}{$\log L_{DCPPM}$} &
    \multicolumn{2}{c}{$\log L_{ILFRS}$} &
    \multicolumn{2}{c}{$\log L_{ILFR}$} \\
    \cmidrule{2-11}
    Dataset & GT & Louvain &
    GT & PPM-max &
    GT & DCPPM-max &
    GT & ILFRS-max &
    GT & ILFR-max\\
    \midrule
    Karate & 0.3715 & 0.4188 & -206 & -191 & -169 & -164 & -176 & -168 & -169 & -160 \\
    Dolphins & 0.3735 & 0.5233 & -483 & -417 & -439 & -399 & -434 & -405 & -429 & -395
     \\
	Football & 0.5877 & 0.6043 & -1384 & -1364 & -1427 & -1407 & -1447 & -1425 & -1428 & -1407
     \\
    Pol-books & 0.4149 & 0.5205 & -1364 & -1183 & -1235 & -1136 & -1285 & -1125 & -1243 & -1089
     \\
    Pol-blogs & 0.4053 & 0.4267 & -73913 & -60654 & -50756 & -49912 & -51843 & -50614 & -50750 & -49702
     \\
    Eu-core & 0.3138 & 0.4211 & -65559 & -57421 & -48783 & -46021 & -49073 & -46202 & -48483 & -45469
     \\
    Cora & 0.5165 & 0.7875 & -678306 & -512556 & -593358 & -450154 & -585837 & -425657 & -584730 & -425463
     \\
    AS & 0.1708 & 0.6307 & -542952 & -449088 & -351537 & -244745 & -338397 & -227192 & -329784 & -227917
     \\
  \bottomrule
\end{tabular}
\end{table*}

We also measured statistical significance of the obtained improvements: for each $\hat\mu$, we applied the algorithms to 5 random samples of LFR and used the paired t-test. For instance, it turns out that according to NMI ILFRS-max never looses significantly and it is significantly better (p-value < 0.05) than: DCPPM-max for $\hat\mu \in \{0.7,0.9,1\}$; DCPPM-iterate for $\hat\mu \in \{0.7,0.8\}$; PPM-max and PPM-iterate for all $\hat\mu \in [0.2,1]$; ILFR-iterate and ILFRS-iterate for $\hat\mu = 0.7$; and Louvain for $\hat\mu \in [0.4,1]$. These results additionally support the choice of the quality function $L_{ILFRS}$~\eqref{eq:ILFRS} for community detection instead of the standard modularity function.

\subsubsection{Real-world complex networks}\label{sec:results_real}
It is argued in several papers that ground truth community labels (a.k.a. metadata) available for some real-world networks should not be used to analyze and compare community detection algorithms~\citep{peel2017ground}. Indeed, such labels are usually obtained using some discrete-valued vertex attributes which are not guaranteed to be reasonable community assignments. For example, users in a social network can be split by gender, city, or interests, and each of these partitions could be treated as ground truth cluster assignments. An algorithm performing better for ``city'' labels can be worse for ``interests'' ones. As a result, no meaningful conclusion can be made based on such comparison.
However, for the sake of completeness, we compare all algorithms on real-world networks, since it is still a standard practice in community detection literature.

Tables~\ref{tab:results_real_max} and~\ref{tab:results_real_iterative} present the results for the iterative and maximization strategies, respectively (in both cases we add Louvain for comparison).\footnote{Recall that Louvain optimizes modularity, so it is similar to DCPPM but with fixed values of $p_{in}$ and $p_{out}$.} Note that we cannot properly measure statistical significance since we cannot sample several copies for each dataset. However, we can account for the randomness included in the algorithms: they all order vertices at each iteration randomly. To do this, for each dataset we run all algorithms 5 times and then apply the unpaired t-test to compare them. We put numbers in bold in Tables~\ref{tab:results_real_max} or~\ref{tab:results_real_iterative} if there is a group of algorithms without significant differences (p-value > 0.05) inside the group and with significant differences with the rest of the algorithms.

The results in general are not consistent: although DCPPM is often the best, the winning model may depend on a graph under consideration and on a target metric. For example, Louvain wins on Cora and AS according to Jaccard index, but it is the worst on the same datasets according to Rand index. 
This supports the claim that ground truth community assignments can be noisy or irrelevant. 

We also noticed that iterative and maximization strategies usually have similar performance and the choice between them is not straightforward. For example, for ILFRS it is often better to apply maximization strategy, while for PPM iterative one should be preferred. Based on this observation, we propose using the iterative strategy, which is faster (it converges after a small number of iteration, as discussed below).

The performance of ILFRS is in general slightly better than of ILFR. Taking into account our results on synthetic networks, we propose using the faster ILFRS instead of ILFR in all practical applications.

Finally, we analyzed the speed of convergence for iterative strategies on real-world datasets. Namely, we measured the number of iterations made by Algorithm~\ref{alg:iterative} before some stopping criteria is reached. Note that the algorithm never stopped because of reaching the maximum allowed number of iterations $N=50$, so we either observed a convergence or applied the cycling criteria (i.e., encountered an already seed parameter).
The average number of iterations obtained for PPM is 12.6 (in all cases the algorithm converged). For DCPPM we got 6.2 (converged in all cases except AS, where we applied cycling stopping criteria after 18 iterations).
For ILFRS we have 5.4 iterations (cycling criteria is applied to 3 datasets).
For ILFR we got 4.4 (again, with 3 applications of cycling criteria).  

\subsection{Statistical Inference, Negative Result}\label{sec:negative}

Having noticed unstable results for real-world datasets, we tried to answer the following research question: \textit{Is any of the null models suitable for community detection in real-world graphs?} 

In order to answer this question, for each quality function we compare its value for the ground truth partition with its value for the partition obtained by the corresponding maximization algorithm (see Table~\ref{tab:compare_likelihood}). For all quality functions, including the widely used modularity, and for all datasets the ground truth partition has a lower value of the quality function. This means that further optimization of \textit{any} quality function would not lead us towards the ground truth partition, which is a negative observation. In particular, there is no hope in improving the results obtained by our algorithms by replacing Louvain-based $FindPartition$ with some more effective maximization algorithm. 

Note that in the literature it is often assumed that a partition with larger modularity is better and, as a result, an algorithm which leads to a partition with larger modularity is better. However, our observation above demonstrates that on the considered real-world datasets it is not the case. 
We also performed additional experiments and noticed that in almost all cases the value of a quality function for the ground truth partition is lower than for the partition obtained by \textit{any} discussed optimization algorithm (not necessary optimizing the same quality function), which is an even stronger negative observation. 

The following conclusion can be made: either the ground truth metadata contained in the considered real-world networks is not a good descriptor for a community structure or statistical inference algorithms based on all null models discussed in this paper are unable to detect real-world communities.

\vspace{-2pt}

\section{Conclusion}\label{sec:conclusion}

In this paper, we focused on the comparison of null models which can be used by likelihood optimization algorithms for community detection. We compared two well-known models, PPM and DCPPM, and a new model, ILFR, which has only one parameter and is proved to preserve the desired expected degree sequence. 
For the new model we have written the log-likelihood, both in parametric and self-contained forms. 
To maximize the parametric log-likelihood functions, we proposed and compared two optimization strategies: maximization and iterative.

The most important conclusion is that the proposed model, ILFR, is the best one for describing most of the considered real-world complex networks according to the likelihood of the observed graph structures, despite the fact that it has only one free parameter. We argue that the likelihood can be considered as the main argument in evaluating the null models instead of the direct comparison of community detection algorithms. The reason is that one cannot fully rely on ground truth cluster labels available for real-world networks. 
Also, we demonstrated that ILFR-based algorithms have the best performance on synthetic networks.
Based on the obtained results, we believe that the proposed ILFR-based quality function~\eqref{eq:ILFRS} can  be successfully used as a target for optimization algorithms, instead of the widely adopted modularity. 


\balance

A natural direction for the future research is to analyze null models for overlapping community detection. This would be useful since many observed networks have overlapping communities. However, fundamental analysis of this problem is complicated by the fact that null models with overlapping communities are less developed and more complex for mathematical analysis.

\vspace{-2pt}




\begin{acks}
This study was funded by RFBR according to the research project 18-31-00207.
\end{acks}

\bibliographystyle{ACM-Reference-Format}
\bibliography{community_detection,hyperparameters} 

\end{document}